\documentclass[letterpaper, 10 pt, conference]{ieeeconf} 
\IEEEoverridecommandlockouts  

\usepackage{amsmath}
\usepackage{amssymb}
\usepackage{amsthm}
\usepackage{mathtools}
\usepackage{derivative}
\usepackage{xcolor}
\usepackage{bm}
\usepackage{cancel}

\usepackage[noadjust]{cite}
\usepackage{url}

\usepackage{graphicx}
\usepackage[export]{adjustbox} 
\graphicspath{{figures/}}

\newtheorem{theorem}{Theorem}
\newtheorem{lemma}{Lemma}

\theoremstyle{definition}
\newtheorem{definition}{Definition}

\newtheorem{remark}{Remark}

\newtheorem{example}{Example}

\newcommand{\R}{\mathbb{R}}

\newcommand{\mc}[1]{\mathcal{#1}}

\newcommand{\Kinfe}{\mathcal{K}_{\infty}^e}
\newcommand{\Ke}{\mathcal{K}^e}

\newcommand{\ReLU}{{\operatorname{ReLU}}}

\newcommand{\bzero}{\mathbf{0}}

\renewcommand{\bf}{\mathbf{f}} 
\newcommand{\bg}{\mathbf{g}}

\newcommand{\bk}{\mathbf{k}}

\newcommand{\bu}{\mathbf{u}}

\newcommand{\bx}{\mathbf{x}}

\usepackage{xcolor}
\definecolor{myblue}{RGB}{49, 114, 174}
\definecolor{myred}{rgb}{0.796, 0.235, 0.2}
\definecolor{mygreen}{rgb}{0.22, 0.596, 0.149}
\definecolor{mypurple}{rgb}{0.584,0.345,0.698}

\usepackage{blindtext}

\title{\textbf{Rectified Control Barrier Functions for High-Order Safety Constraints}}

\author{Pio Ong$^*$, Max H. Cohen$^*$, Tamas G. Molnar, and Aaron D. Ames %
\thanks{PO, MC, and AA are with the Department of Mechanical and Civil Engineering, California Institute of Technology, Pasadena, CA \texttt{\{pioong,maxcohen,ames\}@caltech.edu}.}
\thanks{TM is with the Department of Mechanical Engineering, Wichita State University, Wichita, KS \texttt{\{tamas.molnar\}@wichita.edu}.}
\thanks{$^*$Both authors contributed equally.}
\thanks{This research was supported by the Technology Innovation Institute (TII), Boeing, and NSF CPS Award \#1932091.}
}

\begin{document}
\maketitle
\begin{abstract}
    This paper presents a novel approach for synthesizing control barrier functions (CBFs) from high relative degree safety constraints: \emph{Rectified CBFs (ReCBFs)}. We begin by discussing the limitations of existing High-Order CBF approaches and how these can be overcome by incorporating an activation function into the CBF construction. We then provide a comparative analysis of our approach with related methods, such as CBF backstepping. Our results are presented first for safety constraints with relative degree two, then for mixed-input relative degree constraints, and finally for higher relative degrees. The theoretical developments are illustrated through simple running examples and an aircraft control problem.
\end{abstract}

\section{Introduction}
Control barrier functions (CBFs) \cite{AmesTAC17,AmesECC19} are a useful tool for safety-critical control systems, providing a way to synthesize controllers enforcing state constraints. One of the main advantages of CBFs is the ease of control synthesis using methods such as quadratic programming~\cite{AmesTAC17}, safety filters~\cite{GurrietICCPS18}, or closed-form solutions~\cite{SontagSCL89,PioCDC19}. 
Moreover, there exist numerous extensions of CBFs that address properties beyond safety such as robustness and stability \cite{AmesLCSS19}. Nevertheless, these controller design techniques assume that a CBF is already given for the safety constraint considered.

Constructing CBFs can be challenging, especially when dealing with safety constraints with higher relative degrees. The most popular approach for addressing this issue is High-Order CBFs (HOCBFs)~\cite{WeiTAC22,TanTAC22,SreenathACC16,BreedenAutomatica23,XuAutomatica18}. While effective in some cases, HOCBFs are not traditional CBFs. Many standard results associated with CBFs, such as robustness and stability, do not readily transfer to HOCBFs, and extending these results to HOCBFs is often nontrivial \cite{TanTAC22,XuTAC24,TeelTAC24}.

Another approach to constructing CBFs is backstepping~\cite{AndrewCDC22,AbelTAC23}, which produces CBFs rather than HOCBFs. This method requires systems to be in strict feedback form, or transformable into strict feedback form via output coordinates~\cite{CohenLCSS24}.
Backstepping involves designing a sequence of smooth virtual controllers for a sequence of auxiliary systems, which increases the complexity of control design compared to HOCBFs. 
While recent advancements~\cite{PioCDC19,CohenLCSS23} have made the design of these virtual controllers systematic, the requirements on the system's structure may preclude the application of backstepping to more complex systems.

An advantage of backstepping over HOCBFs is its ability to handle constraints with mixed-input relative degree, in the sense of independent inputs appearing at different orders of derivatives. In the context of HOCBFs, \cite{WeiACC22} addresses this issue using integral control \cite{AmesLCSS21} to dynamically extend inputs, materializing them at different relative degrees. While this enables controller synthesis, it obscures the original inputs in the design process, making it difficult to analyze or minimize control effort.

The main contribution of this paper is the development of a method for constructing CBFs from safety constraints with higher relative degrees. Our approach extends HOCBF ideas by introducing activation functions that consider HOCBF constraints only when necessary. 
The result is a \emph{Rectified Control Barrier Function (ReCBF)}, rather than a HOCBF, that inherits existing properties of CBFs such as stability and robustness. In addition, our approach can generate true CBFs from existing HOCBFs, and it is better suited to handle safety constraints with a weak relative degree for which HOCBF may struggle. We discuss our method by focusing first on safety constraints with relative degree two, and then we move on to mixed-input and higher relative degree constraints.  Moreover, we provide a comparative analysis of ReCBFs with other methods such as HOCBF and backstepping, and illustrate how the main ideas presented herein may also be adapted to these approaches. Finally, we apply our method to a fixed-wing aircraft control problem.
  
\section{Preliminaries and Problem Formulation}
\subsection{Control Barrier Functions}
Consider a nonlinear control affine system\footnote{A continuous function $\alpha\,:\,(-a,b)\rightarrow\R$, $a,b\in\R_{>0}$, is said to be an extended class $\mc{K}$ function ($\alpha\in\mc{K}^e$) if $\alpha(0)=0$ and $\alpha$ is strictly increasing. If $a=b=\infty$ and $\lim_{r\rightarrow\pm\infty}\alpha(r)=\pm\infty$ then $\alpha$ is said to be an extended class $\mc{K}_{\infty}$ function ($\alpha\in\mc{K}^e_{\infty}$).
For a continuously differentiable function $\alpha\,:\,\R\rightarrow\R$, we define ${\alpha'(r)\coloneqq\frac{{\rm d}\alpha}{{\rm d}r}(r)}$. With an abuse of terminology, we say that a function is smooth if it is differentiable as many times as necessary. For a smooth function $h\,:\,\R^n\rightarrow\R^p$ and vector field $\bf\,:\,\R^n\rightarrow\R^n$ we define $L_{\bf}h(\bx)\coloneqq \pdv{h}{\bx}(\bx)\bf(\bx)$ as the Lie derivative of $h$ along $\bf$ with higher order Lie derivatives denoted by $L_{\bf}^{i}h(\bx)\coloneqq \pdv{L_{\bf}^{i-1}h}{\bx}(\bx)\bf(\bx)$. }:
\begin{equation}\label{eq:control-affine-dyn}
    \dot{\bx} = \bf(\bx) + \bg(\bx)\bu,
\end{equation}
with state $\bx\in\R^n$ and input $\bu\in\R^m$, where $\bf\,:\,\R^n\rightarrow\R^n$ and $\bg\,:\,\R^n\rightarrow\R^{n\times m}$ are smooth functions.
Given a locally Lipschitz feedback controller $\bk\,:\,\R^n\rightarrow\R^m$ for  \eqref{eq:control-affine-dyn}, 
the closed-loop system with $\bu=\bk(\bx)$ and initial condition ${\bx_0\in\R^n}$ admits a unique continuously differentiable trajectory $\bx\,:\,I(\bx_0)\rightarrow\R^n$ defined on a maximal interval of existence $I(\bx_0)\subseteq\R_{\geq0}$. Our main objective in this paper is to design feedback controllers $\bk$ such that the closed-loop system satisfies state constraints $\bx(t)\in\mc{C}$ along trajectories, where $\mc{C}\subset\R^n$ is a state constraint set. This is linked to the concept of forward invariance: a set $\mc{S}\subset\R^n$ is said to be forward invariant for the closed-loop system if for each initial condition $\bx_0\in\mc{S}$, the resulting trajectory satisfies $\bx(t)\in\mc{S}$ for all $t\in I(\bx_0)$. While we may wish to design controllers that render the state constraint set $\mc{C}$ forward invariant, such a controller may not exist, and one must instead search for a subset $\mc{S}\subset\mc{C}$ that can be rendered forward invariant. A popular approach to designing controllers enforcing forward invariance of such sets is through CBFs. 

\begin{definition}[\cite{AmesTAC17}]\label{def:cbf}
    A continuously differentiable function $h\,:\,\R^n\rightarrow\R$ defining a set $\mc{S}$ as:
    \begin{equation}\label{eq:S}
        \begin{aligned}
            \mc{S} \coloneqq & \{\bx\in\R^n\,:\,h(\bx)\geq0\}, \\ 
        \end{aligned}
    \end{equation}
    is said to be a CBF for \eqref{eq:control-affine-dyn} on $\mc{S}\subset\R^n$ if there exists $\alpha\in\Ke$ and an open set $\mc{E}\supset\mc{S}$ such that for all $\bx\in\mc{E}$:
    \begin{equation}\label{eq:cbf}
        \sup_{\bu\in\R^m}\left\{L_{\bf}h(\bx) + L_{\bg}h(\bx)\bu \right\} \geq -\alpha(h(\bx)).
    \end{equation}
\end{definition}
The main utility of CBFs is that any locally Lipschitz controller $\bk(\cdot)$ satisfying \eqref{eq:cbf} enforces forward invariance of $\mc{S}$ \cite{AmesTAC17}. In this paper, we focus on constraint sets of the form:
\begin{equation}\label{eq:C}
    \mc{C} \coloneqq \left\{\bx\in\R^n\,:\,\psi(\bx) \geq 0 \right\},
\end{equation}
where ${\psi\,:\,\R^n\rightarrow\R}$ is smooth, and seek CBFs with corresponding zero superlevel sets contained within $\mc{C}\supset\mc{S}$. The following lemma outlines conditions for verifying CBFs. 

\begin{lemma}[\cite{jankovic2018robust}]\label{lemma:cbf}
    A continuously differentiable function $h\,:\,\R^n\rightarrow\R$ is a CBF for \eqref{eq:control-affine-dyn} on a set $\mc{S}$ as in \eqref{eq:S} if and only if there exists $\alpha\in\Ke$ and an open set $\mc{E}\supset\mc{S}$ such that:
    \begin{equation}\label{eq:Lgh=0}
        L_{\bg}h(\bx)=\bzero \implies L_{\bf}h(\bx) \geq - \alpha(h(\bx)),\quad \forall \bx\in \mc{E}.
    \end{equation}
\end{lemma}

\subsection{High-Order Control Barrier Functions}
While Lemma \ref{lemma:cbf} provides a simple condition for verifying a candidate CBF, proposing such a function in the first place is non-trivial for high-dimensional systems where inputs may not directly affect the safety constraint. A popular way to overcome this challenge is via HOCBFs \cite{WeiTAC22,TanTAC22} wherein a candidate CBF is dynamically extended to a new function that may serve as a certificate of safety. The success of this technique relies on the notion of relative degree.
\begin{definition}
    A smooth function $\psi\,:\,\R^n\rightarrow\R$ is said to have relative degree $r\in\mathbb{N}$ for \eqref{eq:control-affine-dyn} at $\bx\in\R^n$ if:
    \begin{enumerate}
        \item $L_{\bg}L_{\bf}^{i}\psi(\bx)=\bzero,\quad \forall i\in\{0,\dots,r-2\},$
        \item $L_{\bg}L_{\bf}^{r-1}\psi(\bx)\neq\bzero.$
    \end{enumerate}
    Similarly, $\psi$ is said to have relative degree $r$ on a set $\mc{E}\subseteq\R^n$ if it has relative degree for all $\bx\in\mc{E}$.
\end{definition}

To define HOCBFs, consider a state constraint set $\mc{C}\subset\R^n$ as in \eqref{eq:C} defined by a smooth function $\psi\,:\,\R^n\rightarrow\R$.
Assuming that $\psi$ has relative degree $r\geq2$ on $\mc{C}$, define:
\begin{equation}\label{eq:psi-HOCBF}
    \psi_{i+1}(\bx) \coloneqq L_{\bf}\psi_{i}(\bx) + \alpha_{i}(\psi_{i}(\bx)),\quad \forall i\in\{0,\dots,r-2\},
\end{equation}
where ${\alpha_i\in\Ke}$ are smooth, with ${\psi_0(\bx)\coloneqq \psi(\bx)}$. This collection of functions produces a collection of sets:
\begin{equation}\label{eq:C-HOCBF}
    \mc{C}_i \coloneqq \left\{\bx\in\R^n\,:\,\psi_i(\bx) \geq 0 \right\},\quad\forall i\in\{0,\dots,r-1\}.
\end{equation}
These sets are used to define a candidate safe set as:
\begin{equation}\label{eq:S-HOCBF}
    \mc{S}\coloneqq \bigcap_{i=0}^{r-1}\mc{C}_i\subset\mc{C},
\end{equation}
which is a subset of the original constraint set $\mc{C} = \mc{C}_0$. The controlled invariance of this safe set can then be ensured through the existence of a HOCBF. 

\begin{definition}[\cite{TanTAC22}]\label{def:hocbf}
    A smooth function $\psi\,:\,\R^n\rightarrow\R$ defining a constraint set $\mc{C}\subset\R^n$ as in \eqref{eq:C} is said to be a HOCBF for \eqref{eq:control-affine-dyn} on a set $\mc{S}\subset\mc{C}$ as in \eqref{eq:S-HOCBF} if there exists an open set $\mc{E}\supset\mc{S}$ and $\alpha\in\Ke$ such that for all $\bx\in\mc{E}$:
    \begin{equation}\label{eq:hocbf}
        \sup_{\bu\in\R^m}\left\{L_{\bf}\psi_{r-1}(\bx) + L_{\bg}\psi_{r-1}(\bx)\bu \right\}\geq - \alpha(\psi_{r-1}(\bx)).
    \end{equation}
\end{definition}
The main result with regard to HOCBFs is that any locally Lipschitz controller satisfying the above condition renders the set $\mc{S}$ forward invariant \cite{WeiTAC22,TanTAC22}. Since $\mc{S}\subset\mc{C}$, this ensures that trajectories remain within the constraint set $\mc{C}$ so long as they are defined. The original definition of a HOCBF \cite{WeiTAC22} does not explicitly require $\psi$ to have relative degree $r$; however, since $L_{\bg}\psi_{r-1}(\bx)=L_{\bg}L_{\bf}^{r-1}\psi(\bx)$, if $\psi$ has relative degree $r$ on $\mc{E}\supset\mc{S}$ then $\psi$ is a HOCBF since $L_{\bg}\psi_{r-1}(\bx)\neq\bzero$ for all $\bx\in\mc{E}$. 
The relative degree requirements of a HOCBF are formalized in \cite{TanTAC22} using the notion of a \emph{weak} relative degree.
\begin{definition}
    A smooth function $\psi\,:\,\R^n\rightarrow\R$ is said to have weak relative degree $r\in\mathbb{N}$ for \eqref{eq:control-affine-dyn} on a set $\mc{E}\subset\R^n$ if it has relative degree $r$ for at least one $\bx\in\mc{E}$ and $L_{\bg}L_{\bf}^{i}\psi(\bx)=\bzero,\, \forall i\in\{0,\dots,r-1\}$ for all other $\bx\in\mc{E}$.
\end{definition}
If $\psi$ has a weak relative degree, 
Lemma \ref{lemma:cbf} may be used to verify HOCBFs: $\psi$ is an HOCBF if $L_{\bf}\psi_{r-1}(\bx)\geq - \alpha(\psi_{r-1}(\bx))$ when $L_{\bg}\psi_{r-1}(\bx)=\bzero$.
Unfortunately, when $\psi$ 
has a relative degree that is weak, it is often not a HOCBF.

\begin{example}[\cite{CohenARC24}]\label{ex:motivating-example}
    Consider a double integrator with state $\bx=(x,\dot{x})\in\R^2$ subject to the following safety constraint:
    $$
    \dot \bx = \begin{bmatrix}
        \dot{x},u
    \end{bmatrix}^\top,\quad \psi(\bx) = 1 - x^2\geq0.
    $$
    By computing $L_{\bg}L_{\bf}\psi(\bx)=-2x$, we have that $\psi$ has relative degree $r=2$ everywhere except when $x=0$. The auxiliary function as in \eqref{eq:psi-HOCBF} is:
    $$\psi_1(\bx)=-2x\dot{x} + \alpha_0(1 - x^2).$$
    For $\psi$ to be a HOCBF, condition \eqref{eq:hocbf} requires $L_{\bf}\psi_{1}(\bx) \geq -\alpha(\psi_1(\bx))$ for all points $\bx\in\mc{S}$
    whenever $L_{\bg}\psi_{1}(\bx)=L_{\bg}L_{\bf}\psi(\bx)=-2x=0$,
    which gives:
    \begin{align*}
    -2\dot{x}^2 - 2\alpha_0'(1 - x^2)x\dot{x} &\geq -\alpha\big(-2x\dot{x}+\alpha_0(1 - x^2)\big)\\
    -2\dot{x}^2 &\geq -\alpha(\alpha_0(1)).
    \end{align*}
    Since $\mc{S} = \{\bx\in\R^2|~\psi(\bx)\geq0,~\psi_1(\bx)\geq0\}$, $\dot{x}$ may take any value when $x=0$ (see Fig. \ref{fig:hocbf}), so we require the inequality above to hold for all $(0,\dot{x})\in\R^2$. Because the right-hand side is constant, there is no $\alpha\in\Ke$ satisfying the inequality for all $(0,\dot{x})\in\mc{S}$, which implies that $\psi$ is not a HOCBF. In particular, when $x=0$ and $|\dot{x}|\geq\sqrt{\alpha(\alpha_0(1))/2}$, there exists no input satisfying \eqref{eq:hocbf}, and, as a result, controllers synthesized using this candidate HOCBF will be ill-defined. For instance, the resulting quadratic programming-based controller (cf. \cite{WeiTAC22}) tends to infinity as $x\rightarrow0$ for $\dot x$ large enough (see Fig.~\ref{fig:hocbf}, right), causing the closed-loop dynamics to exhibit finite escape times (see Fig. \ref{fig:hocbf}, left).
\end{example}

\begin{figure}
    \centering
    \includegraphics{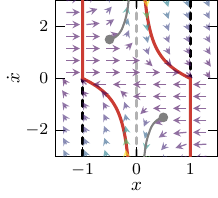}
    \hfill
    \includegraphics{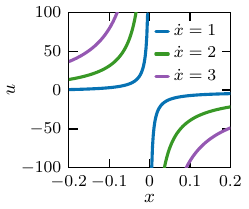}
    \vspace{-0.4cm}
    \caption{\textbf{Left}: Safe set $\mathcal{S}$ induced by the HOCBF candidate from Example \ref{ex:motivating-example}, where the dashed black lines denote the boundary of the constraint set $\mc{C}$, the solid red curves denote the boundary of the safe set $\mc{S}$, the arrows denote the closed-loop vector field under the resulting quadratic programming-based controller (lighter arrows correspond to larger magnitude), the gray curves illustrate example closed-loop trajectories, and the gray dots denote the initial conditions of such trajectories. \textbf{Right}: Input generated by the resulting HOCBF controller for fixed values of $\dot{x}$ as $x$ is varied.}
    \label{fig:hocbf}
\end{figure}

Even when $\psi$ can be verified as a HOCBF, it does not qualify as a CBF in the usual sense. Specifically, the safe set is the zero superlevel set of neither $\psi$ nor $\psi_{r-1}$ but the set intersection defined by \eqref{eq:S-HOCBF}. A limitation of this paradigm is that results for CBFs (e.g., stability and robustness) do not trivially transfer to HOCBFs. In what follows, we present a procedure similar to HOCBFs for constructing CBFs that overcomes these aforementioned limitations.

\section{Rectified Control Barrier Functions}\label{sec:new-backstepping}
\subsection{Weak Relative Degree Two}
The core idea of our approach lies in an activation strategy for HOCBFs. To simplify the discussion and facilitate comparison with other methods, we restrict ourselves to safety constraints  $\psi$ with (weak) relative degree $r=2$ in this section. HOCBFs aim to indirectly render $\psi$ positive along the trajectory by ensuring that $\psi_1(\bx)=L_{\bf}\psi(\bx)+\alpha(\psi(\bx))\geq 0$ along the trajectory, which is achieved by enforcing a CBF-like condition \eqref{eq:hocbf} on $\psi_1$. While such an approach uses the input even when $\psi_1(\bx)\geq 0$, our approach will only invoke the input if necessary,  when $\psi_1(\bx)<0$. 

To this end, we propose the following CBF candidate:
\begin{equation}\label{eq:h-rel-deg-2}
        h(\bx) \coloneqq \psi(\bx) - \ReLU\Big(-\gamma\big(L_{\bf}\psi(\bx) + \alpha(\psi(\bx))\big)\Big),
\end{equation}
with $\ReLU(r)\coloneqq\max\{0,r\}$ the Rectified Linear Unit, continuously differentiable ${\gamma\in\Kinfe}$, ${\gamma'(s)=0\iff s=0}$, and continuously differentiable $\alpha\in\Ke$. Note that one may verify that $\Theta(s) \coloneqq \mathrm{ReLU}(-\gamma(s))$ is continuously differentiable.
The motivation behind \eqref{eq:h-rel-deg-2} is that when the unforced dynamics of \eqref{eq:control-affine-dyn} are safe with $L_{\bf}\psi(\bx) \geq -\alpha(\psi(\bx))$, the second term in \eqref{eq:h-rel-deg-2} is ``deactivated" since it is not required to enforce safety, yielding $h(\bx)=\psi(\bx)$.
We thus refer to \eqref{eq:h-rel-deg-2} as a \textbf{\emph{rectified} CBF (ReCBF)} as higher order terms required to enforce safety are only activated when~$\psi_1(\bx)$ is negative. The following theorem states that, under certain assumptions, ReCBFs are valid CBFs.
\begin{theorem}\label{theorem:new-backstepping-rel-deg-2}
    Consider system \eqref{eq:control-affine-dyn}, a constraint set $\mc{C}\subset\R^n$ defined by a smooth $\psi\,:\,\R^n\rightarrow\R$ as in \eqref{eq:C}, and a set $\mc{S}\subset\mc{C}$ as in \eqref{eq:S} defined by $h\,:\,\R^n\rightarrow\R$ from \eqref{eq:h-rel-deg-2}. If there exists an open set $\mc{E}\supset\mc{S}$ such that $\psi$ has weak relative degree $r=2$ on $\mc{E}$ and:
    \begin{equation}\label{eq:LgLfpsi}
        L_{\bg}L_{\bf}\psi(\bx) = \bzero \implies L_{\bf}\psi(\bx) \geq - \alpha(\psi(\bx)),
    \end{equation}
    for all $\bx\in\mc{E}$, then the ReCBF $h$ is a CBF for \eqref{eq:control-affine-dyn} on $\mc{S}$.
\end{theorem}
\begin{proof}
    We will leverage Lemma \ref{lemma:cbf} to show that $h$ is a CBF. We begin by computing the derivative of $h$ along \eqref{eq:control-affine-dyn}:
    \begin{equation*}
        \begin{aligned}
            \dot{h}
            = & L_{\bf}\psi(\bx) - \Theta'\left(L_{\bf}\psi(\bx) + \alpha(\psi(\bx))\right) \\
            & \times \big(L_{\bf}^2\psi(\bx) + L_{\bg}L_{\bf}\psi(\bx)\bu + \alpha'(\psi(\bx))L_{\bf}\psi(\bx) \big),
        \end{aligned}
    \end{equation*}
    where $\Theta(s) \coloneqq \mathrm{ReLU}(-\gamma(s))$.  From above, we identify:
    \begin{equation*}
        \begin{aligned}
            L_{\bf}h(\bx) = & L_{\bf}\psi(\bx) - \Theta'\left(L_{\bf}\psi(\bx) + \alpha(\psi(\bx))\right)L_{\bf}^2\psi(\bx) \\ & - \Theta'\left(L_{\bf}\psi(\bx) + \alpha(\psi(\bx))\right)\alpha'(\psi(\bx))L_{\bf}\psi(\bx) \\ 
             L_{\bg}h(\bx) = &  -\Theta'\left(L_{\bf}\psi(\bx) + \alpha(\psi(\bx))\right)L_{\bg}L_{\bf}\psi(\bx).
        \end{aligned}
    \end{equation*}
    Thus, $L_{\bg}h(\bx)=\bzero$ if and only if:
    \begin{equation*}
        \Theta'\left(L_{\bf}\psi(\bx) + \alpha(\psi(\bx))\right)=0 \; \vee \; L_{\bg}L_{\bf}\psi(\bx)=\bzero.
    \end{equation*}
    However, since \eqref{eq:LgLfpsi} holds and:
    \begin{equation*}
        \Theta'(s)= \begin{cases}
        0 & \text{if }s \geq 0 \\
        -\gamma'(s) & \text{if } s < 0,
    \end{cases}
    \end{equation*}
    we also have:
    \begin{equation*}
        \begin{aligned}
            L_{\bg}L_{\bf}\psi(\bx) = \bzero \implies & L_{\bf}\psi(\bx) \geq - \alpha(\psi(\bx)) \\ \implies &  \Theta'\left(L_{\bf}\psi(\bx) + \alpha(\psi(\bx))\right)=0,
        \end{aligned}
    \end{equation*}
    which implies that:
    \begin{equation*}
    \begin{aligned}
        L_{\bg}h(\bx) = \bzero \iff & \Theta'\left(L_{\bf}\psi(\bx) + \alpha(\psi(\bx))\right) = 0 \\ 
        \iff & L_{\bf}\psi(\bx) \geq -\alpha(\psi(\bx)).
    \end{aligned}
    \end{equation*}
     Thus, for all $\bx\in\mc{E}$ such that $L_{\bg}h(\bx)=\bzero$, we have $L_{\bf}\psi(\bx) \geq -\alpha(\psi(\bx))$, while the expressions of $h$ and $L_{\bf}h$ yield $h(\bx)=\psi(\bx)$ and $L_{\bf}h(\bx)=L_{\bf}\psi(\bx)$, which leads to:
    \begin{equation*}
    \begin{aligned}
        L_{\bf}h(\bx) = L_{\bf}\psi(\bx) \geq -\alpha(\psi(\bx))= -\alpha(h(\bx)),
    \end{aligned}
    \end{equation*}
    and implies that
    $h$ is a CBF for \eqref{eq:control-affine-dyn} on $\mc{S}$ by Lemma \ref{lemma:cbf}.
\end{proof}
An immediate corollary to the above is that if $\psi$ has relative degree two on a set $\mc{E}\supset\mc{S}$, then the ReCBF $h$ in \eqref{eq:h-rel-deg-2} is a CBF. On the other hand, when the relative degree is weak, condition~\eqref{eq:LgLfpsi} must hold, which is a requirement on the constraint function $\psi$ rather than the auxiliary function $\psi_1$ for HOCBFs.  The controlled invariant set $\mc{S}$ produced by the ReCBF in Theorem~\ref{theorem:new-backstepping-rel-deg-2} is contained within the original constraint set $\mc{C}$ because $\Theta$ is nonnegative and ${h(\bx)\geq0\implies \psi(\bx)\geq0}$. Thus, any controller rendering $\mc{S}$ forward invariant ensures that $\bx(t)\in\mc{C}$ for all $t\in I(\bx_0)$. We conclude this subsection by showcasing the properties of these CBFs compared to other CBF constructions. 

\begin{figure}
    \centering
    \includegraphics{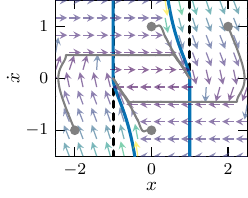}
    \hfill
    \includegraphics{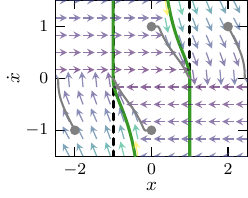}
    \vspace{-0.4cm}
    \caption{\textbf{Left:} Safe set induced by ReCBF \eqref{eq:h-rel-deg-2} for Example \ref{ex:new-cbf} (blue curve), where all other plot elements share the same interpretation as those in Fig. \ref{fig:hocbf}.
    \textbf{Right:} Safe set induced by the CBF \eqref{eq:h-Breeden} for Example \ref{ex:Breeden-cbf} (green curve).}
    \label{fig:new-v-old-cbfs}
\end{figure}

\begin{example}[Comparison to HOCBFs \cite{WeiTAC22}]\label{ex:new-cbf}
    We consider the same scenario as in Example \ref{ex:motivating-example} but now attempt to construct a ReCBF using Theorem \ref{theorem:new-backstepping-rel-deg-2}. Recall that $L_{\bg}L_{\bf}\psi(\bx)=-2x$ and note that $L_{\bf}\psi(\bx)=-2x\dot{x}$ so that $L_{\bg}L_{\bf}\psi(\bx)=0$ implies $x=0$, $L_{\bf}\psi(\bx)=0$, and $\alpha(\psi(\bx))=\alpha(1)$. For $h$ in \eqref{eq:h-rel-deg-2} to be a CBF, \eqref{eq:LgLfpsi} must hold, and it does indeed hold since:
    \begin{equation*}
        L_{\bg}L_{\bf}\psi(\bx)=0 \implies L_{\bf}\psi(\bx) + \alpha(\psi(\bx)) = \alpha(1)\geq0.
    \end{equation*}
    The safe set corresponding to the ReCBF in \eqref{eq:h-rel-deg-2} defined with $\gamma(r)=r|r|$ is plotted with a few example trajectories in Fig. \ref{fig:new-v-old-cbfs} (left), showing safety in accordance with Theorem~\ref{theorem:new-backstepping-rel-deg-2}.
\end{example}

\begin{example}[Comparison to \cite{BreedenAutomatica23}]\label{ex:Breeden-cbf}
    In \cite{BreedenAutomatica23} it is shown that under conditions similar to Theorem \ref{theorem:new-backstepping-rel-deg-2}, the function:
    \begin{equation}\label{eq:h-Breeden}
        h(\bx) = \begin{cases}
            \psi(\bx) & \mathrm{if }\ L_{\bf}\psi(\bx) \geq 0 \\ 
            \psi(\bx) - \frac{1}{2} L_{\bf}\psi(\bx)^2 & \mathrm{if }\ L_{\bf}\psi(\bx) < 0, 
        \end{cases}
    \end{equation}
    is a CBF. A comparison between the zero superlevel sets of the ReCBF $h$ from \eqref{eq:h-rel-deg-2} and CBF from \eqref{eq:h-Breeden} are shown in Fig. \ref{fig:new-v-old-cbfs}, where the sets are almost identical. However, under controllers generated by the ReCBF in \eqref{eq:h-rel-deg-2}, the set $\mc{S}$ is not only forward invariant but also asymptotically stable, as the CBF condition \eqref{eq:cbf} holds not only on $\mc{S}$ but also outside\footnote{In particular, \eqref{eq:cbf} holds for \eqref{eq:h-rel-deg-2} outside of $\mc{S}$ so long as $\psi$ has relative degree $r=2$ outside of $\mc{S}$. On the other hand, even if $\psi$ has relative degree $r=2$ outside of $\mc{S}$, one may show that \eqref{eq:cbf} is violated for \eqref{eq:h-Breeden} at points satisfying $L_{\bf}\psi(\bx)=0$.} of $\mc{S}$ \cite{AmesTAC17}. In contrast, one may show that the CBF condition \eqref{eq:cbf} for \eqref{eq:h-Breeden} does not necessarily hold outside of $\mc{S}$, leading to failure of convergence back to $\mc{S}$. This phenomenon is illustrated in Fig. \ref{fig:new-v-old-cbfs}, where trajectories under ReCBF controllers from \eqref{eq:h-rel-deg-2} stabilize $\mc{S}$ while those corresponding to \eqref{eq:h-Breeden} do not.
\end{example}

\begin{example}[Comparison to backstepping \cite{AndrewCDC22}]\label{ex:backstepping}
Another approach to constructing CBFs is via backstepping \cite{AndrewCDC22,CohenLCSS24}. Here, one considers a safety constraint $\psi$ as in \eqref{eq:C} with weak relative degree $r\geq2$, designs a smooth CBF controller \cite{CohenLCSS23} under the assumption that $\psi$ is a CBF for a single integrator, and then ``backsteps" through this smooth controller to construct a CBF for the original system. More details are available in \cite{AndrewCDC22,CohenLCSS24,CohenARC24}, but for the scenario in Example \ref{ex:motivating-example}, this backstepping CBF is:
\begin{equation}\label{eq:backstepping-cbf}
    h(\bx) = \psi(\bx) - \tfrac{1}{2}(\dot{x} - k(x))^2,
\end{equation}
where $k\,:\,\R\rightarrow\R$ satisfies $\pdv{\psi}{x}(\bx)k(x)\geq - \alpha(\psi(\bx))$. The safe set resulting from this CBF is illustrated in Fig.~\ref{fig:backstepping-cbf} (left) and is shown to be more conservative than the safe set corresponding to the ReCBF from \eqref{eq:h-rel-deg-2} in Fig. \ref{fig:new-v-old-cbfs} (left). However, under appropriate assumptions, the high-level approach in this paper may also be extended to backstepping via taking:
\begin{equation}\label{eq:backstepping-cbf-new}
    h(\bx) = \psi(\bx) - \ReLU\Big(-\gamma\Big(\pdv{\psi}{x}(\bx)\big(\dot{x} - k(x)\big)\Big)\Big),
\end{equation}
with $\ReLU$ as in \eqref{eq:h-rel-deg-2}. The safe set corresponding to this CBF is illustrated in Fig. \ref{fig:backstepping-cbf} (right) and is shown to be similar to that obtained from \eqref{eq:h-rel-deg-2}. While the results in this paper may be extended to backstepping, this would require one to assume \eqref{eq:control-affine-dyn} is in strict feedback form \cite{AndrewCDC22} or that $\psi$ is a function of an output with a valid relative degree \cite{CohenLCSS24}, whereas the current formulation does not require these assumptions. 
\end{example}

\begin{figure}
    \centering
    \includegraphics{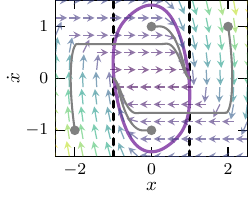}
    \hfill
    \includegraphics{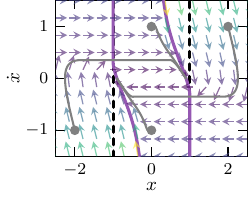}
    \vspace{-0.4cm}
    \caption{\textbf{Left:} Safe set induced by the CBF \eqref{eq:backstepping-cbf} for Example \ref{ex:backstepping} (purple curve), where all other plot elements share the same interpretation as those in Fig. \ref{fig:hocbf}.
    \textbf{Right:} Safe set induced by the CBF \eqref{eq:backstepping-cbf-new} for Example \ref{ex:backstepping}.}
    \label{fig:backstepping-cbf}
\end{figure}

\begin{remark}\label{remark:strict}
    It is often useful to define CBFs \eqref{eq:cbf} with a strict inequality to establish continuity of optimization-based controllers \cite{jankovic2018robust}. The construction in Theorem \ref{theorem:new-backstepping-rel-deg-2} can be modified to produce a CBF satisfying \eqref{eq:cbf} strictly by redefining \eqref{eq:h-rel-deg-2} with $\ReLU(-\gamma(s-\varepsilon))$
    for $\varepsilon>0$, provided \eqref{eq:LgLfpsi} is changed to:
    \begin{equation*}
        L_{\bg}L_{\bf}\psi(\bx) = \bzero \implies L_{\bf}\psi(\bx) \geq - \alpha(\psi(\bx)) + \varepsilon.
    \end{equation*}
\end{remark}
\subsection{Mixed-input relative degrees}
ReCBFs also enable the use of control inputs that appear in higher derivatives of $\psi$ beyond their (weak) relative degree. For example, $\psi$ is a CBF if it has weak relative degree one and
\eqref{eq:Lgh=0} is satisfied:
$$
L_{\bg}\psi(\bx) = \bzero \implies \psi_1(\bx)=L_{\bf}\psi(\bx) +\alpha(\psi(\bx))\geq 0.
$$
When this condition does not hold, the control input may still appear in higher order derivatives of $\psi$, and, unlike HOCBFs, ReCBFs permit the use of such higher order Lie derivatives despite that fact that $L_{\bg}\psi(\bx)\not \equiv \bzero$.

\begin{theorem}\label{theorem:mixed}
    Consider system \eqref{eq:control-affine-dyn}, a constraint set $\mc{C}\subset\R^n$ defined by a smooth $\psi\,:\,\R^n\rightarrow\R$ as in \eqref{eq:C}, and a set $\mc{S}\subset\mc{C}$ as in \eqref{eq:S} defined by $h\,:\,\R^n\rightarrow\R$ from \eqref{eq:h-rel-deg-2}. If there exists an open set $\mc{E}\supset\mc{S}$ such that $L_{\bg}\psi(\bx)$ and $L_{\bg}L_{\bf}\psi(\bx)$ are linearly independent whenever they are nonzero on $\mc{E}$ and:
    \begin{equation*}
        L_{\bg}\psi(\bx)=L_{\bg}L_{\bf}\psi(\bx) = \bzero \implies L_{\bf}\psi(\bx) \geq - \alpha(\psi(\bx)),
    \end{equation*}
    for all $\bx\in\mc{E}$, then the ReCBF $h$ is a CBF for \eqref{eq:control-affine-dyn} on $\mc{S}$.
\end{theorem}
\begin{proof}
    Due to the linear independence assumption, we have $L_{\bg}\psi(\bx)=\bzero$ if $L_{\bg}h(\bx)=\bzero$. By following the proof of Theorem~\ref{theorem:new-backstepping-rel-deg-2}, we get $L_{\bg}h(\bx)=\bzero$ if and only if $L_{\bg}\psi(\bx)=\bzero$ and $L_{\bf}\psi(\bx)\geq-\alpha(\psi(\bx))$, which implies $h$ is a CBF, as in the proof of Theorem \ref{theorem:new-backstepping-rel-deg-2}.
\end{proof}
Theorem~\ref{theorem:mixed} suggests our approach can leverage inputs present in higher-order Lie derivatives when those appearing in lower-order Lie derivatives are insufficient to enforce safety. Similar to backstepping \cite{AndrewCDC22}, this facilitates the synthesis of controllers from mixed relative degree constraints, a situation in which HOCBFs struggle without employing additional techniques such as integral control \cite{WeiACC22}.

\section{Higher Relative Degree ReCBF}
In this section, we extend our results to safety constraints with weak relative degree greater than two.

\begin{definition}
Consider a constraint set~$\mc{C}\subset\R^n$  defined by a smooth $\psi$ as in \eqref{eq:C} with weak relative degree $r\geq2$ on $\mc{E}\subset\R^n$. With $\alpha_{i}(s)\geq\alpha_{i-1}(s)$ for all $s\in \R$, define iteratively:
\begin{subequations}\label{eq:cbf_iteration}
    \begin{align}
\psi_i(\bx) & \coloneqq L_{\bf}h_{i-1}(\bx)+\alpha_{i-1}(h_{i-1}(\bx)), \label{eq:cbf_iteration-psi}\\
h_i(\bx) & \coloneqq h_{i-1}(\bx)-\ReLU(-\gamma_i(\psi_i(\bx))),\label{eq:cbf_iteration-h}
\end{align}
\end{subequations}
for $i \in \mc{I} = \{1,\ldots,r-1\}$, starting with  $h_0(\bx)\coloneqq \psi(\bx)$.
The corresponding \textbf{Rectified CBF (ReCBF)} is defined as: 
\begin{equation}\label{eq:h_general}
h(\bx) \coloneqq h_{r-1}(\bx) = \psi(\bx)- \sum_{i\in \mc{I}}\ReLU(-\gamma_i(\psi_i(\bx))),
\end{equation}
with smooth $\gamma_i\in\Kinfe$, $\alpha_i\in\Ke$, and $\gamma'_i(s)=0\iff s=0$. 
\end{definition}

A ReCBF defines a candidate safe set $\mc{S}$ as in \eqref{eq:S}. Similar to the previous section we have $h(\bx)\geq0\implies \psi(\bx)\geq0$ so that rendering $\mc{S}$ forward invariant ensures satisfaction of the original state constraint. The following result outlines conditions for when a ReCBF $h$ is a valid CBF.
\begin{theorem}\label{theorem:augmented-cbf-rel-deg-r}
    Consider system \eqref{eq:control-affine-dyn}, a constraint set $\mc{C}\subset\R^n$ defined by a smooth $\psi\,:\,\R^n\rightarrow\R$ as in \eqref{eq:C}, and a set $\mc{S}\subset\mc{C}$ as in \eqref{eq:S} defined by a ReCBF $h\,:\,\R^n\rightarrow\R$ from \eqref{eq:h_general}.
    Provided there exists an open set $\mc{E}\supset\mc{S}$ such that $\psi$ has weak relative degree $r$ on $\mc{E}$ and:
    \begin{equation}\label{eq:higher_CBC}
        L_{\bg}L_{\bf}^{r-1}\psi(\bx)=\bzero \implies
        \exists~i \in \mc{I},~\psi_i(\bx)\geq 0,
    \end{equation}
    for all $\bx\in \mc{E}$, then the ReCBF $h$ is a CBF for \eqref{eq:control-affine-dyn} on $\mc{S}$.
\end{theorem}

\begin{proof}
Examining the Lie derivative of $h_i$ along the control directions $\bg$, with $\Theta_i(s) \coloneqq \mathrm{ReLU}(-\gamma_i(s))$:
\begin{equation}\label{eq:Lghi}
    \begin{aligned}
    L_{\bg}h_i(\bx)= & L_{\bg}h_{i-1}(\bx) - \Theta_{i}'(\psi_{i}(\bx))L_{\bg}\psi_{i}(\bx)\\
    =& (1-\alpha_{i-1}'(h_{i-1}(\bx)))L_{\bg}h_{i-1}(\bx)\\ &- \Theta_{i}'(\psi_{i}(\bx))L_{\bg}L_{\bf}h_{i-1}(\bx).
    \end{aligned}
\end{equation}
A similar result also follows when the Lie derivative is taken along the vector field $\bf$, by replacing $\bg$ with $\bf$. Since $L_{\bg}L_{\bf}^i\psi(\bx)\equiv \bzero$ for $i<r-1$ from the weak relative degree assumption, we may ignore the first term with lower order Lie derivative and deduce from repeatedly substituting $h_{i-1}$:
$$
L_{\bg}h(\bx) = (-1)^{r-1}\left(\prod_{i=1}^{r-1}\Theta_i'(\psi_i(\bx))\right)L_{\bg}L_{\bf}^{r-1}\psi(\bx).
$$
To prove that $h$ is a CBF, we appeal to Lemma \ref{lemma:cbf}. Using a similar argument to that in the proof of Theorem \ref{theorem:new-backstepping-rel-deg-2}, \eqref{eq:higher_CBC} implies that $L_{\bg}h(\bx) = \bzero$ if and only if there exists $i \in \mc{I}$ such that $\Theta_i'(\psi_i(\bx))=0$, which occurs when $\psi_i(\bx)\geq 0$ for some $i \in \mc{I}$. Moreover, when $\psi_i(\bx)\geq 0$, we have: 
\begin{align*}
 L_{\bf}h_i(\bx) &= L_{\bf}h_{i-1}(\bx)-\Theta_i'(\psi_i(\bx))L_{\bf}\psi_i(\bx)\\
 &= L_{\bf}h_{i-1}(\bx)\geq -\alpha_{i-1}(h_{i-1}(\bx)) \\
 &= -\alpha_{i-1}(h_{i-1}(\bx)-\Theta_i(\psi_{i}(\bx)))\\
 &= -\alpha_{i-1}(h_{i}(\bx))\geq -\alpha_{i}(h_{i}(\bx)),
\end{align*}
where the inequalities are from the definition of $\psi_i$ and the ReCBF construction of $\alpha_{i}(s)\geq\alpha_{i-1}(s)$. As the above implies that ${\psi_{i+1}(\bx)=L_{\bf}h_i(\bx) + \alpha_i(h_i(\bx))\geq0}$, we may iteratively apply the same procedure to deduce that when $L_{\bg}h(\bx)=\bzero$, we have $\psi_{r}(\bx)=L_{\bf}h_{r-1}(\bx) + \alpha_{r-1}(h_{r-1}(\bx))\geq0$. Since $h(\bx)=h_{r-1}(\bx)$ this implies:
\begin{equation*}
    \begin{aligned}
        L_{\bg}h(\bx)=\bzero \implies & L_{\bf}h(\bx) \geq - \alpha_{r-1}(h(\bx))\geq - \alpha(h(\bx)),
    \end{aligned}
\end{equation*}
for any $\alpha\in\Ke$ satisfying $\alpha(s)\geq\alpha_{r-1}(s)$, which, by Lemma \ref{lemma:cbf}, implies $h$ is a CBF for \eqref{eq:control-affine-dyn} on $\mc{S}$, as desired.
\end{proof}

Theorem~\ref{theorem:augmented-cbf-rel-deg-r} recursively applies a similar methodology to that in Theorem \ref{theorem:new-backstepping-rel-deg-2} to construct a CBF from a safety constraint with an arbitrary weak relative degree. Since ReCBFs are CBFs, results on stability and robustness follow under regular assumptions. Also, similar to Theorem~\ref{theorem:new-backstepping-rel-deg-2}, a corollary to the above result is that \eqref{eq:h_general} is a CBF if $\psi$ has a relative degree on some set $\mc{E}\supset\mc{S}$.

\section{Numerical Examples}
We showcase the main ideas developed herein on an aircraft control problem. We consider simplified pitch dynamics of a fixed-wing aircraft described by \cite{lavretsky2012robust}:
\begin{equation*}
    \underbrace{
    \begin{bmatrix}
        \dot{\theta} \\ \dot{A}_{z}
    \end{bmatrix}}_{\dot{\bx}}
     = 
     \underbrace{
     \begin{bmatrix}
         \frac{g}{V_{T}}\left(A_{z} - \cos(\theta) \right)
         \\ -\frac{1}{\tau}A_{z}
     \end{bmatrix}}_{\bf(\bx)}
     +
     \underbrace{
     \begin{bmatrix}
         0 \\ \frac{1}{\tau}
     \end{bmatrix}}_{\bg(\bx)}
     u
\end{equation*}
where $\theta\in(-\pi,\pi)$ is the pitch angle, $A_{z}\in\R$ is the acceleration along the $z$-axis, $g\in\R_{>0}$ is the gravitational acceleration, $V_{T}\in\R_{>0}$ is the speed of the aircraft (assumed to be fixed), and the input $u\in\R$ denotes the commanded $A_{z}$, which passes through a first-order actuator model characterized by the time-constant $\tau\in\R_{>0}$. Our objective is to design a controller that tracks a prescribed pitch trajectory while enforce symmetric limits on the pitch $|\theta|\leq\theta_{\max}$, captured by the safety constraint $\psi(\bx)=\theta_{\max}^2 - \theta^2$. We verify that this safety constraint satisfies the conditions of our results by first noting that $L_{\bg}L_{\bf}\psi(\bx)=-\frac{2g}{\tau V_{T}}\theta$, implying that $\psi$ has weak relative degree two with $L_{\bg}L_{\bf}\psi(\bx)=0$ when $\theta=0$. Since $L_{\bf}\psi(x) = -2\theta\dot{\theta}$ we have $L_{\bg}L_{\bf}\psi(\bx)=0$ implies that $L_{\bf}\psi(\bx)=0 \geq -\alpha(\theta_{\max}^2)$, implying that the ReCBF $h$ from \eqref{eq:h-rel-deg-2} is a CBF for this system and safety constraint. 

We illustrate the benefits of ReCBFs via comparison to HOCBFs.
To construct a ReCBF we leverage \eqref{eq:h-rel-deg-2} with $\alpha(r)=\tfrac{1}{2}r$ and $\gamma(r)=r|r|$, incorporating $\varepsilon=0.1$ to ensure continuity of the resulting controller (cf. Remark \ref{remark:strict}). This ReCBF is used to construct a safety filter \cite{GurrietICCPS18} that modifies a nominal tracking controller to enforce safety.

The results of our ReCBF safety filter in comparison to an HOCBF safety filter \cite{WeiTAC22} (defined with the same $\alpha$) are illustrated in Fig. \ref{fig:pitch} and Fig. \ref{fig:input}. As seen in Fig. \ref{fig:pitch}, the pitch trajectory generated by ReCBF (blue curve) tracks the desired trajectory (green curve) when it is safe to do so, and prevents the pitch from exceeding its prescribed limits when the desired trajectory leaves the constraint set. The trajectory generated by the HOCBF controller (red curve) initially safely tracks the desired trajectory as well; however, solutions of the closed-loop system fail to exist beyond 10.9 seconds. As shown in Fig. \ref{fig:input}, the control input generated by the HOCBF controller tends to negative infinity as $\theta(t)\rightarrow0$, a point at which $L_{\bg}L_{\bf}\psi(\bx)=0$. Similar to Example \ref{ex:motivating-example}, one may verify that this $\psi$ does not satisfy Def. \ref{def:hocbf} and thus the resulting HOCBF controller is not necessarily well-defined. In contrast, our ReCBF allows trajectories to pass through points where $L_{\bg}L_{\bf}\psi(\bx)=0$, leading to a well-defined controller that handles singularities in $L_{\bg}L_{\bf}\psi(\bx)$.

\begin{figure}
    \centering
    \includegraphics{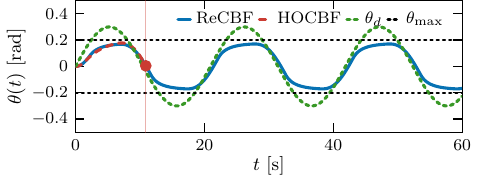}
    \vspace{-0.4cm}
    \caption{Evolution of the pitch angles for different controllers. The blue plot is induced by the ReCBF when used as a safety filter on a nominal control signal that seeks to track the pitch angle shown by the unsafe dotted green line. The red plot, induced by the HOCBF approach, does not have a valid solution after approximately 10.9 seconds.  }
    \label{fig:pitch}
\end{figure}

\begin{figure}
    \centering
    \includegraphics{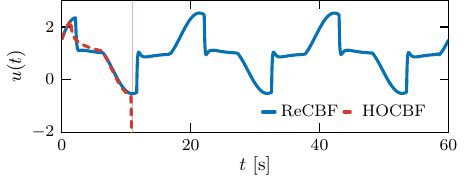} 
    \vspace{-0.4cm}
    \caption{Input signals for the ReCBF and HOCBF approach for the aircraft example. The HOCBF input goes unbounded as $L_{\bg}L_{\bf}\psi(\bx)\rightarrow0$.}
    \label{fig:input}
\end{figure}

\section{Discussion and Conclusions}
This paper introduces \emph{Rectified CBFs}: a tool for constructing CBFs for high relative degree constraints that overcomes limitations posed by 
traditional techniques, such as HOCBFs.
We provided detailed technical treatments for three scenarios: (i) relative degree two safety constraints, (ii) constraints where independent inputs affect derivatives of varying orders up to two and (iii) higher relative degree constraints.
We presented a comparative analysis of our approach with existing approaches. While our method offers some theoretical advantages over HOCBFs by handling constraints with weak relative degrees, it is not without its own limitations. The controllers generated by ReCBFs are sensitive to the various hyperparameters on which they depend, and improper tuning of these hyperparameters can lead to controllers with large Lipschitz constants that produce large input. Thus, characterizing the properties of these controllers in relation to their hyperparameters is an important direction for future work. Other future research directions include unifying our results on mixed and high relative degrees.

\bibliographystyle{ieeetr}
\bibliography{biblio}

\end{document}